\documentclass[manyauthors]{fundam}

\usepackage{lineno,hyperref}

\usepackage{amsmath,amssymb}
\usepackage{tikz}
\usepackage{graphics}
\usepackage{multirow}
\usepackage{gastex}
\newcommand{\ra}{\rightarrow}
\def\l{[\ }
\def\r#1{{{\ ]_{}}_{}}_{#1}}

\begin{document}

\setcounter{page}{243}
\publyear{2021}
\papernumber{2072}
\volume{182}
\issue{3}

 \finalVersionForARXIV
%% \finalVersionForIOS

\title{Time-free Solution to Independent Set Problem \\
                          using P Systems with Active Membranes}

\author{Yu Jin\\
School of Mathematical Sciences\\
Zhejiang University\\
 Hangzhou 310058, Zhejiang, China\\
jinyu907@hotmail.com
\and Bosheng Song\\
College of Information Science and Engineering\\
Hunan University\\
 Changsha 410082, Hunan, China\\
boshengsong@hnu.edu.cn
\and Yanyan Li\\
College of Information Science and Engineering\\
Hunan University\\
 Changsha 410082, Hunan, China\\
yanyanli@hnu.edu.cn
\and Ying Zhu\thanks{Address for correspondence: Hunan Construction Advanced Technical College, Changsha 410015, Hunan, China.\newline \newline
          \vspace*{-6mm}{\scriptsize{Received  June 2021; \ revised August 2021.}}}
\\
Hunan Construction Advanced Technical College\\
 Changsha 410015, Hunan, China\\
532657761@qq.com
 }

\maketitle

\runninghead{Y. Jin, et al.}{Time-free Solution to Independent Set Problem using P Systems with Active Membranes}

\begin{abstract}
{Membrane computing is a branch of natural computing aiming to abstract computing models from the structure and functioning of living cells.
The computation models obtained in the field of membrane computing are usually called P systems.
P systems have been used to solve computationally hard problems efficiently on the assumption that
the execution of each rule is completed in exactly one time-unit
(a global clock is assumed for timing and synchronizing the execution of rules).
However, in biological reality, different biological processes take different times to be completed,
which can also be influenced by many environmental factors. In this work, with this biological reality,
we give a time-free solution to independent set problem using P systems with active membranes,
which solve the problem independent of the execution time of the involved rules.}
\end{abstract}


\section{Introduction}\label{sec:Intro}
\emph{Membrane computing} is a computing paradigm motivated by the structure and functioning of the living cells.
It was initiated by Gh. P\u{a}un in \cite{puaun2000computing,puaun2000membrane} and has developed rapidly (already in 2003, the Institute for Scientific Information, ISI,
declared membrane computing as ``fast emerging research area in computer science'', see http://esi-topics.com).
The computation devices considered in the framework of membrane computing are usually called \emph{P systems}.
There are three main classes of P systems investigated:
cell-like P systems \cite{puaun2000computing,jiang2019computational}, tissue-like P systems \cite{martin2003tissue,song2021computational,song2021monodirectional,song2020monodirectional}, neural-like P systems \cite{ionescu2006spiking,zhang2014weighted,carandang2019handling,lazo2021return}.
In recent years, many variants of spiking neural P systems have been considered \cite{pan2009spiking,pan2010spiking,pan2011spiking,pan2011time,pan2011spiking2}.
For general information in membrane computing, one may consult \cite{gh2002membrane,paun2010handbook,song2021survey,cturlea2019search,zandron2020bounding}
and to the membrane computing website \url{http://ppage.psystems.eu/} for details.

\medskip
\emph{P systems with active membranes} are a class of cell-like P systems, which were introduced in \cite{paun2001computing}. Generally speaking,
P systems with active membranes consist of a hierarchical structure composed by $q$ membranes, where the outermost membrane is called the \textit{skin} membrane. Membranes delimit \emph{regions}, that contain some objects
(represented by symbols of an alphabet),
the region outside the membrane system is called the \emph{environment}.
A feature of these systems is the fact that the membranes are polarized,
they have one of the three possible ``electrical charges'': positive ($+$),
negative ($-$) or neutral ($0$). The whole membrane structure,
the charge of membranes and the objects contained in membranes evolve by using the following types of rules:
(a) object evolution; (b) object communication; (c) membrane dissolution;
(d) membrane division. Usually, the rules are applied in a nondeterministic and maximally parallel way (any object and membrane
which can evolve by a rule of any form, should evolve).

\medskip
P systems with active membranes
have been successfully used to solve computationally hard problems efficiently \cite{alhazov2003solving,alhazov2004trading,pan2011computation,pan2005solving,pan2006further,puaun2004power}.
All these above mentioned P systems with active membranes work in a parallel and synchronized way (a global clock is assumed to mark the time for the system), in each tick of the global clock, all the applicable rules are applied simultaneously, and the execution of rules takes exactly one time unit. However, in biological reality, different biological processes take different times to be completed,
which can also be influenced by many environmental factors\cite{valencia2020simulation}. Thus, a \emph{timed P system} was proposed in \cite{cavaliere2005time}, to each rule a natural number representing the execution time of the rule is associated.
A particular class of timed P systems is called \emph{time-free P systems},
such P systems produce always the same result, independent from the execution times of the rules.

\medskip
The notion of time-free solution to hard computational problems was introduced in \cite{cavaliere2012timefree}.
In \cite{song2014time}, time-free solution to {\tt SAT} problem using P systems with active membranes was present,
where the computation result is independent of the execution time of the involved rules.
Although independent set problem can be reduced to {\tt SAT} problem in polynomial time,
it remains open how to compute the reduction by P systems.
In this work, we give a direct solution to independent set problem using P systems with active membranes,
instead of computing the reduction by P systems.
The solution to independent set problem using P systems with active membranes
is time-free in the sense that the computation result is independent of the execution time of the involved rules.

The organization of this paper is described as follows. Section
\ref{sec:activeP} presents some fundamental conceptions of language and automata theory and
the notion of timed P systems with active membranes. A time-free solution to independent set problem by P systems with active
membranes is investigated in Section
\ref{sec:tf-solution}. Finally, conclusions and
some future works are given in Section \ref{sec:Con}.

\section{P systems with active membranes}\label{sec:activeP}
\subsection{Preliminaries}
It is useful for the reader to have some familiarity with notion and notation from formal language theory \cite{rozenberg1997handbook},
as well as the definition of P systems with active membranes \cite{gh2002membrane}.

For an alphabet $V$, $V^*$ denotes the set of all finite strings of symbols from $V$,
while the empty string is denoted by $\lambda$, and the set of all non-empty strings over $V$ is denoted by $V^+$.
The length of a string $x$ is denoted by $|x|$ and by $card(A)$ the cardinality of the set $A$.

By $\mathbb{N}$ we denote the set of non-negative integers.
A multiset over an alphabet $V=\{a_1,a_2,\dots,a_n\}$ is a mapping $m$: $V\rightarrow \mathbb{N}$.
We can represent a multiset $m$ over $V$ as any string $w\in V^*$ such that $|w|_{a_i}=m(a_i),1\leq i\leq n$.
That is, $m(w)=(m(a_1),m(a_2),\dots,m(a_n))$.
We usually represent $m$ by the string $a_1^{m(a_1)}\dots a_k^{m(a_k)}$ or by any permutation of this string.

\subsection{P systems with active membranes}

\begin{definition}
A \emph{P system with active membranes} of degree $m$ is a construct
$$\Pi=(O,H,\mu,w_1,\dots,w_m,R),$$
where:
\begin{enumerate}
\item[(i)] $m\ge 1$ is the degree of the system;
\item[(ii)] $O$ is the alphabet of {\it objects};
\item[(iii)] $H$ is a finite set of {\it labels} for membranes;
\item[(iv)] $\mu$ is the initial {\it membrane structure}, consisting of $m$ membranes; membranes are labelled in an injective way with elements of $H$ and are electrically polarized, being possible charge positive ($+$), negative ($-$) or neutral ($0$);
\item[(v)] $w_1,\dots,w_m$ are strings over $O$, describing the {\it
initial multisets of objects} placed in the $m$ regions of $\mu$;
\item[(vi)] $R$ is a finite set of {\it development rules}, of the following types:
\begin{enumerate}
\item[(a)] $\l a\ra v\r h^\alpha$, $h\in H,\alpha\in\{+,-,0\}, a\in O,
v\in O^*$.\\ (object evolution rules, associated with membranes
and depending on the label and the charge of the membranes);

\item[(b)] $a\l \r h^{\alpha_1}\ra \l b \r h^{\alpha_2}$, $h\in H,
\alpha_1,\alpha_2\in\{+,-,0\},a,b\in O$.\\ (send-in rules; an object
is sent into the membrane, possibly modified during this process; also the
polarization of the membrane can be modified, but not its label);

\item[(c)] $\l a \r h^{\alpha_1}\ra \l \r h^{\alpha_2}b$, $h\in H,
\alpha_1,\alpha_2\in\{+,-,0\},a,b\in O$.\\ (send-out rules; an object
is sent out of the membrane, possibly modified during this process; also the
polarization of the membrane can be modified, but not its label);

\item[(d)] $\l a \r h^\alpha\ra b$, $h\in H,\alpha\in\{+,-,0\},
a,b\in O$.\\ (dissolving rules; in reaction with an object, a membrane can be
dissolved, while the object specified in the rule can be modified);

\item[(e)] $\l a \r h^{\alpha_1}\ra \l b \r h^{\alpha_2} \l c
\r h^{\alpha_3}$, $h\in H,\alpha_1,\alpha_2,\alpha_3\in\{+,-,0\},a,b,c\in O$.\\ (division
rules for elementary membranes; in reaction with an object, the membrane is
divided into two membranes with the same label, possibly of different polarizations;
the object specified in the rule is replaced in the two new membranes by possibly new objects;

\item[(f)] $ [ [\ ]_{h_1}^{\alpha_1}\dots \l  \ ]_{h_k}^{\alpha_1}\
[\ ]_{h_{k+1}}^{\alpha_2}\dots [\ ]_{h_n}^{\alpha_2}
\r {h_0}^{\alpha_0}\ra [ [\ ]_{h_1}^{\alpha_3}\dots
[ \ ]_{h_k}^{\alpha_3}\r {h_0}^{\alpha_5}\ [ [\ ]_{h_{k+1}}^{\alpha_4}\dots
[ \ ]_{h_n}^{\alpha_4}]_{h_0}^{\alpha_6}$, $k\ge 1, n>k,h_i\in H,
0\le i\le n$, and $\alpha_0,\dots,\alpha_6\in\{+,-,0\}$ with $\{\alpha_1,\alpha_2\}
=\{+,-\}$.\\ (if the membrane with label $h_0$ contains other membranes than
those with the labels $h_1,\dots,h_n$ specified above, then they must have
neutral charges; these membranes are duplicated and then are part of the contents of both new copies of the
membrane $h_0$).
\end{enumerate}
\end{enumerate}
\end{definition}

The above rules can be considered as ``standard'' rules of P systems with active membranes;
the following two rules can be considered as the extension of rules (a) and (e), respectively.

\begin{enumerate}
\item[(a$'$)] $\l u\ra v\r h^\alpha$, $h\in H,\alpha\in\{+,-,0\}, u,v\in O^*$.\\ (cooperative evolution rules, associated with membranes
and depending on the label and the charge of the membranes);
\item[(e$'$)] $\l a \r h^{\alpha}\ra \l a_1 \r {h_1}^{\alpha_1}\l a_2 \r {h_2}^{\alpha_2}\dots \l a_d \r {h_d}^{\alpha_d}$, $h,h_1,\dots,h_d\in H$, $\alpha,\alpha_1,\dots, \alpha_d\in\{+,-,0\}$, $a,a_1,\dots$, \\ $a_d\in O$, $d\geq 2$.\\
($h$ is an elementary membrane; in reaction with an object, the membrane is
divided into $d$ membranes not necessarily with the same label; also the polarizations of the new membranes
can be different from the polarization of the initial one;
the object specified in the rule is replaced in the $d$ new membranes by possibly new objects).
\end{enumerate}

For a detailed description of using these rules
we can refer to \cite{paun2001computing,gh2002membrane}.
Here, we mention that the rules are used in the non-deterministic maximally parallel
manner, and we assume that the rules are applied in the bottom-up manner:
in any given step, one uses first the evolution rules of type (a), (a$'$),
then the other rules which also involve a membrane;
moreover, one uses first the rules of types (b), (c), (d), (e), (e$'$)
and then those of type (f). We also remark that
at one step a membrane $h$ can be subject of only one rule of types (b)-(f) and (e$'$).
A configuration in a P system with active membranes is described by the membrane structure,
together with charge on each membrane and the multisets of objects in each region.
A P system with active membranes evolves from one configuration to the next one by applying rules as mentioned above.
A sequence of transitions between configurations defines a computation.
A computation halts if it reaches a configuration where no rule can be applied in any membrane.
The result of a computation is the multiset of objects contained into an output membrane, or
emitted from the skin of the system.

\subsection{Timed P systems with active membranes}
The notion of timed P system was proposed from \cite{cavaliere2005time}. In this work,
we consider timed P systems with active membranes.

A \emph{timed P system with active membranes} $\Pi(e)=(O,H,\mu,w_1,\dots,w_m,R,e)$ is obtained
by adding a time-mapping $e: R\rightarrow \mathbb{N}$ to
a P system with active membranes $\Pi=(O,H,\mu,w_1,\dots,w_m,R)$,
where $\mathbb{N}$ is the set of natural numbers and the time-mapping
$e$ specifies the execution times for the rules.

A timed P system with active membranes $\Pi(e)$ works in the following way.
An external clock is assumed, which marks time-units of equal length, starting from instant 0.
According to this clock, the step $t$ of computation is defined by the period of time that goes from instant $t-1$ to instant $t$.
If a membrane $i$ contains a rule $r$ from types (a) -- (f), (a$'$) and (e$'$) selected to be executed,
then execution of such rule takes $e(r)$ time units to complete.
Therefore, if the execution is started at instant $j$, the rule is completed at instant $j + e(r)$
and the resulting objects and membranes become available only at the beginning of step $j + e(r) + 1$.
When a rule $r$ from types (b) -- (f) and (e$'$) is applied, then the occurrences of symbol-objects and the membrane subject to this rule
cannot be subject to other rules from types (b) -- (f) and (e$'$) until the implementation of the rule completes.
At one step, a membrane can be subject to several rules of types (a) and (a$'$).

\subsection{Recognizer timed P systems with active membranes}

In this subsection, we first present the definition of recognizer P systems with active membranes, then the notion of recognizer timed P systems with active membranes is given.

\begin{definition}
A recognizer P system with active membranes of degree $m\geq 1$ with input is a tuple $\Pi=(O,H,\Sigma,\mu,$ $w_1,\dots,w_m,R,i_{out},i_{in})$, where:
\begin{itemize}
\item The tuple $(O,H,\mu,w_1,\dots,w_m,R,i_{out})$ is a P system with active membranes.
\item $\Sigma$ is an (input) alphabet strictly contained in $O$.
\item The initial multisets $w_1, \dots , w_m$ are over
  $O\setminus\Sigma$.
\item $i_{in} \in \{1, \dots ,m\}$ is the label of a distinguished
  (input) membrane.
\item The working alphabet contains two
distinguished elements {\tt yes} and {\tt no}.
\item All the computations halt.
\item If $\mathcal C$ is a computation of the system, then either
  object {\tt yes} or object {\tt no} (but not both) must appear in
  the environment when the system halts. Note that object {\tt yes} or object {\tt no} can be present in a non-halting configuration.
\end{itemize}
\end{definition}

For recognizer P systems with active membranes, we say that a computation is an {\em accepting computation} (resp., {\em rejecting computation}) if the object {\tt yes} (resp., {\tt no}) appears in the environment associated with the corresponding halting configuration.

Here, differently from the usual interpretation, we allow yes and no objects to exit into the environment before reaching the halting configuration. In that case they are not providing the answer to the decision problem.

For each multiset $w$ over the input alphabet $\Sigma$, the {\it
computation of P systems with active membranes $\Pi$ with input} $w$ starts from the
configuration of the form
$(w_1,\dots,w_{i_{in}} +
w,\dots, w_m, \mu)$,
that is, the input multiset $w$
has been added to the contents of the input membrane
$i_{in}$. Therefore, we have an initial configuration associated with
each input multiset $w$ (over the input alphabet $\Sigma$) in this
kind of systems.

\begin{definition}
A recognizer timed P system with active membranes of degree $m\geq 1$ is a tuple $(\Pi,e)$,
where $\Pi$ is a recognizer P system with active membranes and $e$ is a time-mapping of $\Pi$.
\end{definition}

\subsection{Time-free solutions to decision problems by P systems with active membranes}
In this subsection, we introduce the definition of time-free solutions to decision problems by P systems with active membranes \cite{song2014time}.

In a timed P systems with active membranes,
a computation step is called a {\em rule starting step} (RS-step, for short)
if at this step at least one rule starts its execution.
In the following, we will only count RS-steps as the definition of
time-free solutions to decision problems by P systems with active membranes
(i.e., steps in which some object ``starts'' to evolve or some membrane ``starts'' to change).
In timed P systems with active membranes,
the execution time of rules is determined by the time mapping $e$,
and it is possible the existence of rules whose execution time is inherently exponential,
therefore, the number of RS-steps in a computation characters how ``fast" the constructed
P system with active membranes solves a decision problem in the context of time-freeness.

\begin{definition}
Let $X=(I_X, \Theta_X)$ be a decision problem. We say that $X$ is solvable in a {\em time-free polynomial time}
by a family of recognizer P systems with active membranes $\Pi=\Pi_u, u \in I_X$
(we also say that the family $\Pi$ is a {\em time-free solution} to the decision problem $X$) if the following items are true:
\begin{itemize}
\item the family $\Pi$ is polynomially uniform by a Turing machine;
that is, there exists a deterministic Turing machine working in polynomial time which constructs the system $\Pi_u$ from the instance $u \in I_X$.
\item the family $\Pi$ is \emph{time-free sound} (with respect to $X$);
that is, for any time-mapping $e$, the following property holds:
if for each instance of the problem $u \in I_X$ such that there
exists an accepting computation of $\Pi_u(e)$, we have $\Theta_X(u)=1$.
\item the family $\Pi$ is \emph{time-free complete} (with respect to $X$);
that is, for any time-mapping $e$, the following property holds:
if for each instance of the problem $u \in I_X$ such that $\Theta_X(u)=1$,
every computation of $\Pi_u(e)$ is an accepting computation.
\item the family $\Pi$ is \emph{time-free polynomially bounded}; that is,
there exists a polynomial function $p(n)$ such that for any time-mapping $e$ and for each $u\in I_X$,
all computations in $\Pi_u(e)$ halt in, at most, $p(|u|)$ RS-steps.
\end{itemize}
\end{definition}

\section{A time-free solution to independent set problem by P systems with active membranes}\label{sec:tf-solution}

In this section, we first introduce the definition of independent set problem,
then construct a family of P systems with active membranes that solve independent set problem in a time-free polynomial time.

\medskip
\textbf{Independent Set Problem}

\smallskip
INSTANCE: A undirected graph $\gamma=(V,E)$, where $V=\{v_1,v_2,\dots,v_n\}$ is the set of vertices,
$E$ is the set of edges with elements of the form $(v_i,v_j)$, $v_i,v_j\in V$, $i\neq j$,
and a positive integer $k<card(V)$.

\smallskip
QUESTION: Is there a subset $V'\subseteq V$ with $card(V')\geq k$ such that no two vertices in $V'$ are jointed by an edge in $E$?

\begin{theorem}\label{Th:indepset}
Independent set problem can be solved by a family of P systems with active membranes
in a time-free polynomial time with respect to the size of the problem.
\end{theorem}

\begin{proof} Let us consider a undirected graph $\gamma=(V,E)$,
where $V=\{v_1,v_2,\dots,v_n\}$ is the set of vertices,
$E$ is the set of edges with elements of the form $(v_i,v_j)$,
$v_i,v_j\in V$, $i\neq j$, and a positive integer $k<card(V)$.

For the given undirected graph $\gamma$, suppose that the undirected graph $\gamma$ has $s$ ($s\leq (n^2-n)/2$) edges which are ordered,
we construct the P systems with active membranes
$$\Pi_{\gamma}=(O,H,\mu,w_0,w_1,w_{n+3+s}, R),$$
where
\begin{itemize}
  \item $O=\{v_i,v_i',v_i'',g_i,g'_i\mid 1\le i\le n\}\cup\{a_i,e_i\mid 1\le i\le s\}\cup \{{\tt yes}, {\tt no},a,a_{s+1},b_0,b,c,d,d'$, $d'',d'''\}\cup\{b_i \mid 1\leq i\leq s+1\}$ is the alphabet,
  \item $H=\{-1,0,1,2,\dots,n+3+s\}$ is the set of labels of the membranes,
  \item $\mu=\l \l \l \r 1^0 \r 0^0 \r {n+3+s}^0$ is initial membrane structure,
  \item $w_0=\lambda$ (that is, membrane 0 contains no object in the initial configuration),
  \item $w_1=\{b,v_1,v_2,\dots, v_n\}$ is the initial multiset contained in membrane 1,
  \item $w_{n+3+s}=\{{\tt no}\}$ is the initial multiset contained in membrane $n+3+s$,
\end{itemize}
and the set $R$ contains the following rules (we also give explanations about the role of these rules in the computation of solving independent set problem):

\medskip
$r_{1,i}: \l v_i \r i^0\rightarrow \l v_i' \r {i+1}^0 \l v_i'' \r {i+1}^0$, $1\leq i\leq n$.

$r_{2,i}: \l \!v_i'\rightarrow e_{h_{i,1}}\dots e_{h_{i,j_{i}}} g_i\r {i+1}^0$, $1\leq\! i\!\leq n$, and each edge $e_{h_{i,1}},\dots, e_{h_{i,j_{i}}}$ is linked with vertex$\;v_i$.

$r_{3,i}: \l v_i''\rightarrow cg'_i \r {i+1}^0$, $1\leq i\leq n$.

$r_{4,i}: \l g_i \r {i+1}^0\rightarrow g_i \l \r {i+1}^+$, $1\leq i\leq n$.

$r_{5,i}: \l g'_i \r {i+1}^0\rightarrow g'_i \l \r {i+1}^-$, $1\leq i\leq n$.

$r_{6,i}: \l \l \r {i+1}^+ \l \r {i+1}^-\r 0^0\rightarrow \l \l \r {i+1}^0\r 0^0\l \l \r {i+1}^0\r 0^0$, $1\leq i\leq n$.

\medskip
At step 1, the rule $r_{1,1}: \l v_1 \r 1^0\rightarrow \l v_1' \r 2^+ \l v_1'' \r 2^-$ is applied,
producing the objects $v_1'$ and $v_1''$, which are placed in two separate copies of membrane 2.
Note that when the membrane with label 1 is divided by the rule $r_{1,1}$, the obtained two membranes have label 2 instead of label 1.
For any given time-mapping $e$, the execution of rule $r_{1,1}$ completes in $e(r_{1,1})$ steps.
As we will see below, at step 1, except for the application of rule $r_{1,1}$,
rule $r_{14}: \l {\tt no} \r {n+3+s}^0 \rightarrow\l \r {n+3+s}^+ {\tt no}$ also starts;
and from step 2 to step $e(r_{1,1})$, there is no rule starting.
Thus, during the execution of rule $r_{1,1}$ (i.e., from step 1 to step $e(r_{1,1})$),
there is one RS-step (here the rule $r_{14}$ does not count). Note that the number of RS-step during the execution of rule
$r_{1,1}$ is independent from the time-mapping $e$.

\medskip
After the execution of rule $r_{1,1}$ completes, the application of rule
$r_{2,1}: \l v_1'\rightarrow e_{h_{1,1}}\dots e_{h_{1,j_{1}}}g_1 \r 2^0$ and rule
$r_{3,1}: \l v_1''\rightarrow c g'_1\r 2^0$ starts. Note that the application starts at the same step,
but they may complete at different steps. For any given time-mapping $e$,
the execution of rule $r_{2,1}$ and rule $r_{3,1}$ takes one RS-step.

When the execution of rule $r_{2,1}$ (resp., $r_{3,1}$) completes, rule $r_{4,1}$ (resp., $r_{5,1}$) starts to use,
a dummy object is sent out of the membrane, the charge is change to positive (resp., negative). Note that the executive of rules $r_{2,1}$ and $r_{3,1}$
may start at different steps.

\medskip
Rule $r_{6,1}:\l \l \r 2^+ \l \r 2^-\r 0^0\rightarrow \l \l \r 2^0\r 0^0\l \l \r 2^0\r 0^0$
can be applied only when the execution of rule $r_{4,1}$ and rule $r_{5,1}$ completes.
For any given time-mapping $e$, the execution of rule $r_{6,1}$ takes $e(r_{6,1})$ steps, where the number of RS-step is one.

By the application of rule $r_{6,1}$, the polarization of the membranes with label 2 changes to neutral.
In this way, the rule $r_{1,2}: \l v_2 \r 2^0\rightarrow \l v_2' \r 3^+ \l v_2'' \r 3^-$ is enabled and applied.
Similar to the case of vertex $v_1$, the process of execution of vertex $v_2$ takes five RS-steps,
and four membranes with label 0 are generated, each membrane with label 0 contains a membrane with label 3.
In general, after $5n$ RS-steps, $2^n$ separate copies of membrane with label 0 are generated,
all of which are placed in the membrane with label $n+3+s$; each membrane with label 0 contains a membranes with label $n+1$ (see Fig. \ref{fig:membrane-struc}).

\begin{figure}[htb]
\vspace{1mm}
	\centering
	\centerline{\unitlength 0.8mm}
		%\fbox{
		\begin{tikzpicture}
		
		\node [ font=\fontsize{16}{16}\selectfont] at (10.7,3.5) {...};
		\draw [rounded corners] (15.7,1.5)node [below ]{$n+3+s$} -- (5.4,1.5) -- (5.4,5.2) --(15.7, 5.2)-- cycle;
		\draw [rounded corners] (15,2.3)node [below ]{$0$} -- (11.5,2.3) -- (11.5,4.6) --(15, 4.6)-- cycle;
		\draw [rounded corners] (10,2.3)node [below ]{$0$} -- (6.2,2.3) -- (6.2,4.6) --(10, 4.6)-- cycle;
		\draw [rounded corners] (14,3)node [below ]{$n+1$} -- (12,3) -- (12,4) --(14, 4)-- cycle;
		\draw [rounded corners] (9,3)node [below ]{$n+1$} -- (7,3) -- (7,4) --(9, 4)-- cycle;
		\end{tikzpicture}
	%}
	\caption{The membrane structure of the system $\Pi_{\gamma}$ after $3n$ RS-steps}
	\label{fig:membrane-struc}
\end{figure}
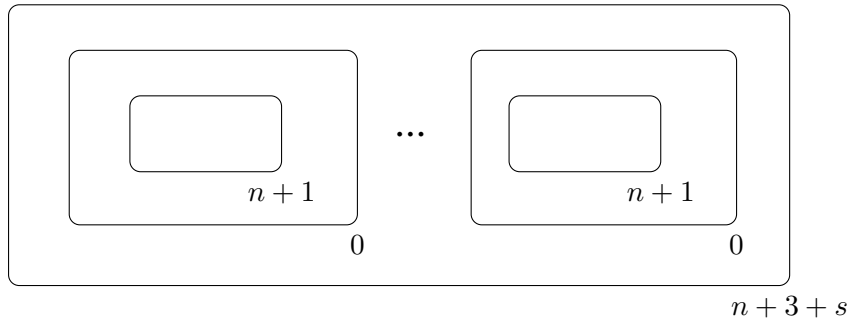

\iffalse
\begin{figure}[htb]
\centering
\centerline{\unitlength 0.8mm
%\fbox{
\begin{tikzpicture}(100,33)(55,81)
\gasset{Nw=18,Nh=10,Nmr=4}
\node(t1)(76,101){}  \nodelabel[ExtNL=y,NLdist=0.1,NLangle=-30](t1){$n+1$}
\node[Nw=35,Nh=20](t2)(80,100){}  \nodelabel[ExtNL=y,NLdist=1,NLangle=-30](t2){$0$}
\node[Nw=18,Nh=10](t3)(126,101){}  \nodelabel[ExtNL=y,NLdist=1,NLangle=-30](t3){$n+1$}
\node[Nframe=n](t0)(105,100){$\cdots$}
\node[Nw=35,Nh=20](t4)(130,100){}  \nodelabel[ExtNL=y,NLdist=1,NLangle=-30](t4){$0$}
\node[Nw=95,Nh=32](t5)(106,99){}  \nodelabel[ExtNL=y,NLdist=1,NLangle=-20](t5){$n+3+s$}
\end{tikzpicture}}%
%}
\caption{The membrane structure of the system $\Pi_{\gamma}$ after $3n$ RS-steps}
\label{fig:membrane-struc}
\end{figure}
\fi

$r_7: \l b \r {n+1}^0\rightarrow \l b_0 \r {n+2}^+ \l b_1\r {n+2}^0\dots \l b_s\r {n+1+s}^0\l b_{s+1}\r {n+2+s}^0$.

$r_8: \l b_0\rightarrow aa_1\dots a_s a_{s+1}e^2_1\dots e^2_s c^{n-k+1} \r {n+2}^+$.

$r_9: \l a \r {n+2}^+\rightarrow \l \r {n+2}^0 a$.

$r_{10,i}: \l e^2_i\rightarrow d \r {n+1+i}^0$, $1\leq i\leq s$.

$r_{11,i}: \l d \r {n+1+i}^0\rightarrow \l \r {n+1+i}^- d$, $1\leq i\leq s+1$.

$r_{12,i}: \l a_i \r {n+1+i}^-\rightarrow \l d''\r {n+2+i}^0 \l d'''\r {-1}^0$, $1\leq i\leq s$.

$r_{13}: \l c^{n-k+1}\rightarrow d \r {n+2+s}^0$.

$r_{14}: \l d^{s+2}\rightarrow d' \r 0^0$.

$r_{15}: \l d' \r 0^0\rightarrow d'$.

\medskip
Each neutral membrane with label $n+1$ containing object $b$ is produced at the same time.
By applying the rule $r_7$, one positive membrane with label $n+2$ and $s+1$ neutral membranes
with labels $n+2+i$ ($0\leq i\leq s$) are produced, respectively.
After the execution of rule $r_7$ completes, the application of rules $r_8$, $r_{10,i}$
(if at least two copies of object $e_i$ exist in membranes $n+1+i$ ($1\leq i\leq s$) which are generated by rule $r_7$) and $r_{13}$
(if there are at least $n-k+1$ copies of object $c$ in membrane $n+2+s$ which is generated by rule $r_5$) starts at the same step,
but they may complete at different steps. When the execution of rule $r_8$ completes,
the execution of rule $r_9: \l a \r {n+2}^+\rightarrow \l \r {n+2}^0 a$ starts,
where object $a$ exits the membrane, changing its polarization from positive to neutral.
After the execution of rule $r_9$, the application of rules $r_{10,i},r_{11,i}$ and $r_{12,i}$ will be applied one by one
(at this time, the evolution objects which are generated by rule $r_8$, for the rule $r_{12,i}$,
the membrane containing object $a_i$ is divided into two membranes with label $n+2+i$ and $-1$, respectively,
where the membrane with label $-1$ is a ``dummy" membrane that will not evolve anymore).
In this way, the rules of types $r_{10,i},r_{11,i},r_{12,i}$ are applied as many times as possible.
At some moment, when the membrane with label $n+2+s$ is generated by the rule $r_{12,s}$,
the application of rule $r_{13}: \l c^{n-k+1}\rightarrow d \r {n+2+s}^0$ starts,
$n-k+1$ copies of object $c$ evolve to object $d$. After the execution of rule $r_{13}$,
the application of rule $r_{11,s+1}: \l d \r {n+2+s}^0\rightarrow \l \r {n+2+s}^- d$ starts,
object $d$ exits the membrane, changing its polarization from neutral to negative.

It is important to note that when the execution of all rules $r_9,r_{10,i},r_{11,i},r_{12,i},r_{13}$
(the evolution objects which are generated by rule $r_8$) completes,
the execution of rules $r_{10,i},r_{11,i},r_{13}$ (these rules are enabled due to the previous application of rule $r_7$,
i.e., the membranes used for rules $r_{10,i},r_{11,i},r_{13}$ are generated by rule $r_7$) must already complete.

\medskip
If $s+2$ copies of object $d$ are present in membrane 0, the application of rule
$r_{14}: \l d^{s+2}\rightarrow d' \r 0^0$ starts, object $d'$ is produced
(it means there are $s+1$ copies of object $d$ which are evolved from $e^2_1,e^2_2,\dots,e^2_s,c^{n-k+1}$ (generating by rule $r_8$),
and one copy of $d$ comes from one of the membranes with label $n+2+i$ $(0\leq i\leq s)$ which are produced by rule $r_7$).
After the execution of rule $r_{14}$, the application of rule
$r_{15}: \l d' \r 0^0\rightarrow d'$ starts, where the membrane with label 0 is dissolved.

Note that we need to check that we have at least $k$ selected vertices,
hence we dissolve the membranes with label 0 only when less than $k$ vertices are marked with a prime.

\medskip
For any given time-mapping $e$, the execution of rule $r_7$ completes in $e(r_7)$ steps,
where there is one RS-step; the execution of rules $r_8$ and $r_{10,i}$
(these rules are applied in the membranes with labels $n+2+i$ $(0\leq i\leq s)$
which are produced by rule $r_7$) takes one RS-step; the execution of rules $r_{11,i}$
(object $d$ is generated by the membrane $n+2+i$ $(0\leq i\leq s)$, which is produced by rule $r_7$) takes at most $s+1$ RS-steps;
the execution of rules $r_9,r_{10,i},r_{11,i},r_{12,i}$ ($1\leq i\leq s$)(the evolution objects which are generated by rule $r_8$) takes $3s+1$ RS-steps; the execution of rules $r_{13},r_{14},r_{15},r_{11,s+1}$ takes four RS-steps;
thus, the total number of RS-steps is $4s+8$.

\medskip
$r_{16}: \l {\tt no} \r {n+3+s}^0 \rightarrow\l \r {n+3+s}^+ {\tt no}$.

$r_{17}: {\tt no}\l \r {n+3+s}^-\rightarrow \l {\tt no} \r {n+3+s}^-$.

$r_{18}: d'\l \r 0^0 \rightarrow \l {\tt yes} \r 0^0$.

$r_{19}: \l {\tt yes} \r 0^0\rightarrow \l \r 0^0 {\tt yes}$.

$r_{20}: \l {\tt yes} \r {n+3+s}^+\rightarrow \l \r {n+3+s}^- {\tt yes}$.

\medskip
At step 1, the rule $r_{16}: \l {\tt no} \r {n+3+s}^0 \rightarrow\l \r {n+3+s}^+ {\tt no}$ is applied,
object ${\tt no}$ exits the skin membrane $n+3+s$, changing its polarization from neutral to positive.

\medskip
When the execution of rule $r_{15}$ completes,
if no membrane with label 0 is present in the skin membrane with label $n+3+s$,
then the rules $r_{18}: d'\l \r 0^0 \rightarrow \l {\tt yes} \r 0^0$,
$r_{19}: \l {\tt yes} \r 0^0\rightarrow \l \r 0^0 {\tt yes}$ and
$r_{20}: \l {\tt yes} \r {n+3+s}^+\rightarrow \l \r {n+3+s}^- {\tt yes}$ cannot be applied,
thus, rule $r_{17}: {\tt no}\l \r {n+3+s}^-\rightarrow \l {\tt no} \r {n+3+s}^-$ cannot be applied.
In this case, when the computation halts, object {\tt no} remains in the environment,
telling us that there is not a subset $V'\subseteq V$ with $card(V')\geq k$ such that
no two vertices in $V'$ are jointed by an edge in $E$.
Note that the system will take computation steps to complete the execution of rule $r_{16}$,
but there is no RS-step from this moment to the end of the execution of rule $r_{16}$.

\medskip
When the execution of rule $r_{15}$ completes,
if some membranes with label 0 still exist,
then the rule $r_{18}$ will be applied, where object $d'$ evolves to {\tt yes},
and object {\tt yes} enters the membrane.
When the execution of rule $r_{18}$ completes, the application of rule $r_{19}$ starts,
object {\tt yes} exits the membrane with label 0.
At this moment, if the execution of rule $r_{16}$ is not yet completed,
then no rule can be started in the system before the execution of rule $r_{16}$ completes.
Only when the execution of rule $r_{16}$ completes,
the polarization of membrane $n+3+s$ changes to positive,
and the rule $r_{20}$ is enabled and applied.
By applying the rule $r_{20}$, object {\tt yes} exits the membrane with label $n+3+s$,
changing its polarization from positive to negative. Therefore, the other objects {\tt yes}
remaining in membrane $n+3+s$ are not able to continue exiting into the environment.
After the execution of rule $r_{20}$ completes, the rule $r_{17}$ is enabled and applied,
object {\tt no} enters membrane $n+3+s$. In this case, when the computation halts,
one copy of {\tt yes} appears in the environment,
telling us that there is a subset $V'\subseteq V$ with $card(V')\geq k$ such that
no two vertices in $V'$ are jointed by an edge in $E$.

\medskip
It is clear that for any time-mapping $e: R \ra \mathbb{N}$,
the object {\tt yes} appears in the environment
when the computation halts
if and only if the independent set exists;
and the object {\tt no} stays in the environment
when the computation halts
if and only if the independent set does not exist.
So, the system $\Pi_{\gamma}$ is time-free sound and time-free complete.

For any time-mapping $e: R \ra \mathbb{N}$, if the independent set exists, the computation takes at most $5n+4s+12$ RS-steps, the system halts. If the independent set does not exist, the computation takes at most $5n+4s+8$ RS-steps, and the system halts. Thus, the family ${\bf \Pi}$ is time-free polynomially bounded.

\medskip
The family ${\bf \Pi}= \{\Pi_{\gamma} \mid \gamma \textrm{ is an instance of independent set problem}\}$ is polynomially uniform
because the construction of P systems described in the proof can be done in maximum time (polynomial) by a Turing machine:
\begin{itemize}
\itemsep=0.95pt
\item the total number of objects is $5n+3s+12$;
\item the number of initial membranes is $3$;
\item the cardinality of the initial multisets is $3$;
\item the total number of evolution rules is $6n+3s+12$;
\item the maximal length of a rule (the number of symbols necessary to write a rule, both its left and right sides, the membranes, and the polarizations of membranes involved in the rule) is $n+3s-k+6$.
\end{itemize}

Thus, independent set problem can be solved in a time-free polynomial RS-steps
with respect to the size of the problem by recognizer P systems with active membranes and this concludes the proof.
\end{proof}

\section{Conclusions and remarks}\label{sec:Con}

In this work, with the biological reality: different biological processes take different times to be completed, which can
also be influenced by many environmental factors,
we give a time-free solution to independent set problem using P systems with active membranes,
which solve the problem independent from the execution time of the involved rules.

The notion of ``time-free solutions to decision problems by P systems with active membranes" was given in section \ref{sec:activeP},
it is possible that the execution time of a rule is inherently exponential with respect to the size of an instance,
thus, a more reasonable definition was given, we use the RS-steps to character how ``fast" the constructed P system with active membranes solves a decision problem in the context of time-freeness.

The solution to independent set problem in this work is semi-uniform in the sense that P systems are constructed from the instances of the problem. It remains open how can we construct a uniform time-free solution to independent set problem (that is, a P system can solve a family of instances of the same size). In section \ref{sec:tf-solution}, P systems constructed in the proof of Theorem \ref{Th:indepset} have the rules of types (a$'$), (b), (c), (d), (e$'$) and (f), it remains open whether the rule types used in this construction can be weakened, for instance whether changing the labels of the membranes created via rules of type (e$'$) is actually necessary or if the construction can be carried on without changing any of the membrane labels.

%Channel states were considered in cell-like P systems with symport/antiport rules in , where
%objects exchanged between neighbor regions are controlled by channel states. It is of interest to
%investigate the computational power of P systems with rule production and removal that objects exchanged between neighbor regions
%are controlled both by channel states and target indications.
%
%The flat maximal parallelism of using rules was considered in \cite{Freund,Verlan} and
%this way of using rules was developed in \cite{pan,song2016flat}, where
%in each step, a maximal set of applicable rules is chosen and
%each rule in the set is applied exactly once in each region of a system. What is the computational
%power of RPR P systems by using rules in a flat maximally parallel way?
%
%In \cite{macias2015,pan2010computational}, membrane separation was introduced into (tissue) P systems with symport/antiport rules
%and the computational complexity was investigated, where the length of symport/antiport rules is considered as
%a parameter for the computation power.
%It is interesting to investigate the computational complexity of RPR P systems with membrane separation.

\subsection*{Acknowledgements}
This work was supported by National Natural Science Foundation of
China (61972138), the Fundamental Research Funds for the Central Universities
(531118010355), Hunan Provincial Natural Science Foundation of China (2020JJ4215), and
the Key Research and Development Program of Changsha (kq2004016).

%\nocite{*}
\bibliographystyle{fundam}

%%%%%%%%%%%%%%%%%%%%%%%%%%%%%%%%%%%%%%%%%%%%%%%%%%%%%%%%%%%%%%%%%%%%%%

\end{document}